\documentclass[a4paper,UKenglish]{lipics}
\usepackage{doc}
\usepackage[utf8]{inputenc}
\usepackage{amsmath}
\usepackage{amsfonts}
\usepackage{amsthm}
\usepackage{amssymb}
\usepackage{alltt}
\usepackage{bussproofs}
\usepackage{color}
\usepackage{csquotes}
\usepackage{dsfont}
\usepackage{float}
\usepackage{varwidth}
\usepackage{wasysym}
\usepackage{mathtools}
\usepackage{multirow}
\usepackage{verbatim}
\usepackage{csquotes}
\usepackage{lipsum}
\usepackage{fancyvrb}
\usepackage{hyperref}
\usepackage{rotating}
\usepackage{tikz}
 
\usepackage{microtype}

\newcommand{\Olist}[2]{\left[ #1 \middle| #2 \right] }
\newcommand{\CRlist}[3]{\left[ #1 \middle| #2 \middle| #3 \right] }
\newcommand{\CRmid}[1]{\widehat{ #1 }}

\title{A Generalized Resolution Proof Schema and the Pigeonhole Principle}
\titlerunning{A Generalized Resolution Proof Schema and the Pigeonhole Principle} 

\author{David M. Cerna}

\affil{Research Institute for Symbolic
Computation (RISC) \\ Johannes Kepler University, Linz, Austria\\
\texttt{David.Cerna@risc.jku.at}}

\authorrunning{D.\,M. Cerna } 

\Copyright{David M. Cerna}


\serieslogo{}
\volumeinfo
  {Billy Editor and Bill Editors}
  {2}
  {Conference title on which this volume is based on}
  {1}
  {1}
  {1}
\EventShortName{}
\DOI{10.4230/LIPIcs.xxx.yyy.p}
\begin{document}
\maketitle

\pagestyle{plain}

\begin{abstract}
The {\em schematic CERES method} is a method of cut elimination for {\em proof schemata}, that is a sequence of proofs with a recursive construction. Proof schemata can be thought of as a way to circumvent the addition of an induction rule to the \textbf{LK}-calculus. In this work, we formalize a schematic version of the {\em Infinitary Pigeonhole Principle} (IPP), in the \textbf{LKS}-calculus \cite{CERESS2}, and analyse the extracted clause set schema. However, the refutation we find cannot be expressed as a resolution proof schema \cite{CERESS2} because there is no clear ordering of the terms indexing the recursion, every ordering is used in the refutation. Interesting enough, the clause set and its refutation is very close to a ``canonical form'' found in cut elimination of \textbf{LK}-proofs~\cite{SimonThesis}. Not being able to handle refutations of this form is problematic in that proof schema, when instantiated, are \textbf{LK}-proofs. Based on the structure of our refutation and structural results~\cite{SimonThesis}, we develop a generalized resolution proof schema based on recursion over a special type of list, and provide a refutation, using our generalization, of the clause set extracted from our formal proof of IPP. We also extract a Herbrand System from the refutation.
\end{abstract}

\section{Introduction}\label{sec:Introduction}

In Gentzen's {\em Hauptsatz}~\cite{Gentzen1935}, a sequent calculus for first order logic was introduced, namely, the \textbf{LK}-calculus. He then went on to show that the {\em cut} inference rule is redundant and in doing so, was able to show consistency of the calculus. The method he developed for eliminating cuts from \textbf{LK}-derivations works by inductively reducing the cuts in a given \textbf{LK}-derivation to cuts which either have a reduced {\em formula complexity} and/or reduced {\em  rank}~\cite{prooftheory}. This method of cut elimination is known as {\em reductive cut elimination}. A useful result of cut elimination for the \textbf{LK}-calculus is that cut-free \textbf{LK}-derivations have the {\em subformula property},  i.e. every formula occurring in the derivation is a subformula of some formula in the end sequent. This property allows for the construction of {\em Herbrand sequents} and other objects which are essential in proof analysis. 

Using cut elimination, it is also possible to gain mathematical knowledge concerning the connection between different proofs of the same theorem. For example, Jean-Yves Girard's application of reductive cut elimination to a variation of F\"{u}rstenberg-Weiss' proof of Van der Waerden's theorem ~\cite{ProocomWaerdens1987} resulted in the {\em analytic} proof of Van der Waerden's theorem as found by Van der Waerden himself. From the work of Girard, it is apparent that interesting results can be derived from eliminating cuts in ``mathematical'' proofs. 

A more recently developed method of cut elimination, the CERES method ~\cite{CERES}, provides the theoretic framework to directly study the cut structure of \textbf{LK}-derivations, and in the process reduces the computational complexity of deriving a cut-free proof. The cut structure is transformed into a clause set allowing for clausal analysis of the resulting clause form. Methods of reducing clause set complexity, such as {\em subsumption} and {\em tautology elimination} can be applied to the characteristic clause set to reduce its complexity. It was shown by Baaz \& Leitsch in ``Methods of cut Elimination''~\cite{Baaz:2013:MC:2509679} that this method of cut elimination has a {\em non-elementary speed up} over reductive cut elimination.

The CERES has been used to analyze connections between proofs well~\cite{Baaz:2008:CAF:1401273.1401552}. The method was applied to F\"{u}rstenberg's proof of the infinitude of primes and the resulting clause set contained Euclid's argument for prime construction. 

A mathematical formalization of  F\"{u}rstenberg's proof requires induction. In the higher-order formalization, induction is easily formalized as part of the formula language. However in first-order, an induction rule needs to be added to the \textbf{LK}-calculus. As it was shown in ~\cite{CERESS2}, reductive cut elimination does not work in the presence of an induction rule in the \textbf{LK}-calculus. Also, other systems~\cite{Mcdowell97cut-eliminationfor} which provide cut elimination in the presence of an induction rule do so at the loss of some essential properties, for example the subformula property.   

In ``Cut-Elimination and Proof Schemata''~\cite{CERESS2}, a version of the \textbf{LK}-calculus was introduced (\textbf{LKS}-calculus) allowing for the formalization of sequences of proofs as a single object level construction, i.e. the {\em proof schema}, as well as a framework for performing cut elimination on proof schemata. Cut elimination performed within the framework of  ~\cite{CERESS2} results in cut-free proof schemata with the subformula property. Essentially, the concepts found in ~\cite{CERES} were generalized to handle recursively defined proofs. It was shown in ~\cite{CERESS2} that {\em schematic} characteristic clause sets are always unsatisfiable, but it is not known whether a given schematic characteristic clause set will have a refutation expressible as a resolution proof schema. This gap distinguishes the schematic version of the CERES method from the previously developed versions. 

The method of~\cite{CERESS2} was used in~\cite{OrderPigeonsArxivVersion,MyThesis} to perform an analysis of a schema representing the ordered infinitary pigeonhole principle. a refutation of the clause set was formalized in the resolution proof schema of~\cite{CERESS2} and a Herbrand system was extracted.  In this work, we consider the infinitary pigeonhole principle which has been referred to in literature as the  {\em tape proof}, found in  ~\cite{TAPEPROOFNOEQ,tapeproofpaper,TapeproofEX2}. We  generalize the tape proof by considering a codomain of size $n$ rather than of size two, of  which we refer to as the {\em Non-injectivity  Assertion} (NiA-schema).

While analysing the NiA-schema using the schematic CERES method we ran into problems because the refutation of the clause set we found in Sec. \ref{sec:refuteset} cannot be formalized as a resolution proof schema. It requires every ordering of the $\omega$-terms indexing the refutation, while the definition of resolution proof schema requires a specific ordering. A solution would be to generalize the definition of resolution proof schema, but a generalization based on this particular example will not hold much weight when concerning general proof schema. However, the results of Sec. 6 of~\cite{SimonThesis} (\textbf{TACNF} normal form) concerning clause sets extracted at various stages of applying Gentzen style cut-elimination to a proof, are closely related to a clause set we derive in Sec.~\ref{sec:refuteset} and our refutation in Sec.~\ref{genref}. As long as one does not remove atomic cuts, the clause sets derived from various stages of Gentzen style cut-elimination create a subsumption hierarchy with a canonical form of clause set at the bottom. We develop our generalization of the resolution proof schema such that it follows the refutation of the canonical form of clause set at the bottom of the hierarchy. Also, our generalization retains the subformula property. We generalize resolution proof schema  by performing recursion over {\em carriage return list} (see Sec. \ref{sec:CRLGSRC}) rather than  over numerals. As an end result, we provide a refutation of the  NiA-schema's clause set in the generalized resolution proof schema and extract its Herbrand System.
 
The paper is structured as follows: In Sec. \ref{sec:SCERES}, we introduce the \textbf{LKS}-calculus and the essential concepts from~\cite{CERESS2}. In Sec~\ref{sec:MathNiA} \& \ref{sec:FormNiA}, we formalize the NiA-schema in the \textbf{LKS}-calculus. In Sec.~\ref{sec:CCSSE}, we extract the characteristic clause set from the NiA-schema and perform {\em normalization} and tautology elimination. In Sec. \ref{sec:refuteset}, we provide a (``mathematically defined'')  refutation proof schema. In Sec.~\ref{sec:CRLGSRC} we introduce the concept of carriage return list and generalized  refutation proof schema. We then provide a formalization of the NIA-schema's refutation in the new refutation proof schema definition and extract a Herbrand system.   In Sec. \ref{sec:Conclusion}, we conclude the paper and discuss future work. 

\section{The \textbf{LKS}-calculus and Clause set Schema}\label{sec:SCERES}

In this section we introduce the necessary background material from~\cite{CERESS2} such as the  \textbf{LKS}-calculus, clause set schema, resolution schema and Herbrand systems. 

\subsection{Schematic language, proofs, and the \textbf{LKS}-calculus}
The \textbf{LKS}-calculus is a schematic version of \textbf{LK}-calculus constructed by Gentzen~\cite{Gentzen1935}. A proof in the \textbf{LKS}-calculus has an indexing {\em parameter}, which, when instantiated,  results is an  \textbf{LK}-derivation~\cite{CERESS2}. We extend the term language to accommodate schematic constructs necessary for \textbf{LKS}-derivations. We work in a two-sorted setting containing a {\em schematic sort} $\omega$ and an {\em individual sort} $\iota$. The schematic sort contains numerals constructed from the constant $0:\omega$, a monadic function $s(\cdot):\omega \rightarrow \omega$ as well as $\omega$-variables $\mathcal{N}_{v}$ (introduced in~\cite{OrderPigeonsArxivVersion}), of which one variable, the {\em free parameter}, will be used to index \textbf{LKS}-derivations. The parameter will be represented by $n$ unless otherwise noted. 

The individual sort is essentially a standard first order term language~\cite{prooftheory}, but we allow schematic function symbols. Thus,  $\iota$ contains countably many constant symbols, countably many {\em constant function symbols}, and  {\em defined function symbols}. The constant function symbols are standard terms and the defined function symbols are used for schematic terms. Though, it is allowed to have defined function symbols unroll into numerals and thus, can be of type $\omega^n \to \omega$.  The $\iota$ sort also has {\em free} and {\em bound} variables and an additional concept, {\em extra variables}~\cite{CERESS2}. These are variables introduced during the unrolling of defined function ({\em predicate}) symbols. Also important are the {\em schematic variable symbols} which are variables of type $\omega \rightarrow \iota$. Essentially second order variables, though, when evaluated with a {\em ground term} from the $\omega$ sort we treat them as first order variables. Our terms are built inductively using constants and variables as a base.  

Formulae are constructed inductively using countably many {\em predicate constants}, logical operators $\vee$,$\wedge$,$\rightarrow$,$\neg$,$\forall$, and $\exists$, as well as  {\em defined predicate symbols} which are used to construct schematic formulae, similar to defined function symbols. In this work {\em iterated $\bigvee$} is the only defined predicate symbol used. Its formal specification is:
\begin{equation}
\label{eq:one}
\varepsilon_{\vee}= \bigvee_{i=0}^{s(y)} P(i) \equiv \left\lbrace \begin{array}{cccc}
{\displaystyle \bigvee_{i=0}^{s(y)} P(i) \Rightarrow \bigvee_{i=0}^{y} P(i) \vee P(s(y)) } & & &
{\displaystyle \bigvee_{i=0}^{0} P(i) \Rightarrow P(0)}
\end{array}\right\rbrace 
\end{equation}
Using the above term and formula language we define the \textbf{LKE}-calculus,  the \textbf{LK}-calculus~\cite{prooftheory} plus an equational theory $\varepsilon$ (in our case $\varepsilon_{\vee}$ Eq. \ref{eq:one}). The equational theory is a primitive recursive term algebra describing the structure of the defined function (predicate) symbols. The \textbf{LKS}-calculus is the  \textbf{LKE}-calculus with {\em proof links}.
\begin{definition}[$\varepsilon$-inference rule]
\begin{prooftree}
\AxiomC{$S\left[ t\right] $}
\RightLabel{$(\varepsilon)$}
\UnaryInfC{$S\left[ t'\right] $}
\end{prooftree}
In the $\varepsilon$ inference rule, the term $t$ in the sequent $S$ is replaced by a term $t'$ such that, given the equational theory  $\varepsilon$,  $\varepsilon \models t = t'$.
\end{definition}

To extend the \textbf{LKE}-calculus with  proof links we need a countably infinite set of {\em proof symbols}  denoted by $\varphi, \psi,\varphi_{i}, \psi_{j} \ldots$. Let $S(\bar{x})$ by a  sequent with a vector of schematic variables $\bar{x}$, by  $S(\bar{t})$ we denote the sequent $S(\bar{x})$ where each of the variables in $\bar{x}$ is replaced by the terms in the vector $\bar{t}$ respectively, assuming that they have the appropriate type. Let $\varphi$ be a proof symbol and $S(\bar{x})$ a sequent, then the expression \AxiomC{$(\varphi(\bar{t}))$}
\dashedLine
\UnaryInfC{$S(\bar{t})$}
\DisplayProof
is called a {\em proof link} . For a variable $n:\omega$, proof links
such that the only $\omega$-variable is $n$ are called {\em $n$-proof links} \index{k-proof Link}.

\begin{definition}[\textbf{LKS}-calculus~\cite{CERESS2}]
The sequent calculus $\mathbf{LKS}$
consists of the rules of $\mathbf{LKE}$, where proof links may appear
at the leaves of a proof.
\end{definition}

\begin{definition}[Proof schemata~\cite{CERESS2}]\label{def.proofschema}
\index{Proof Schemata}
  Let $\psi$ be a proof symbol and $S(n,\bar{x})$ be a sequent
  such that $n:\omega$. Then a {\em proof schema pair for $\psi$} is a pair of $\mathbf{LKS}$-proofs $(\pi,\nu(k))$ with end-sequents $S(0,\bar{x})$ and $S(k+1,\bar{x})$ respectively such that $\pi$ may not contain proof links and $\nu(k)$ may
  contain only proof links of the form \AxiomC{$(\psi(k,\bar{a}))$}
  \dashedLine
  \UnaryInfC{$S(k,\bar{a})$}
  \DisplayProof,
 we say that it is a proof link to $\psi$. We call $S(n,\bar{x})$ the end sequent of $\psi$ and assume an identification between the formula occurrences in the end sequents of $\pi$ and $\nu(k)$ so that we can speak of occurrences in the end sequent of $\psi$. Finally a proof schema $\Psi$ is a tuple of proof schema pairs for $\psi_1 , \cdots \psi_\alpha$ written as $\left\langle \psi_1 , \cdots \psi_\alpha \right\rangle$, such that the $\mathbf{LKS}$-proofs for $\psi_{\beta}$ may also contain $n$-proof links to $\psi_{\gamma}$ for $1\leq \beta < \gamma\leq \alpha$. We also say that the end sequent of $\psi_1$ is the end sequent of $\Psi$. 
\end{definition}

For more information concerning proof schemata and the calculus we  refer the reader to~\cite{CERESS2}. We now move on to the {\em characteristic clause set schema}.

\subsection{Characteristic Clause set Schema}
Extraction of a characteristic clause set from an \textbf{LK} proof (see CERES method~\cite{CERES}) required inductively following the formula occurrences of cut formula ancestors up the proof tree to the leaves. In proof schemata, the concept of ancestors and formula occurrence is more complex. A formula occurrence might be an ancestor of a cut formula in one recursive call and in another it might not. Additional machinery is necessary to extract the characteristic clause term from proof schemata. A set $\Omega$ of formula occurrences from the end-sequent of an \textbf{LKS}-proof $\pi$ is called {\em a configuration for $\pi$}. A configuration $\Omega$ for $\pi$ is called relevant w.r.t. a proof schema $\Psi$ if $\pi$ is a proof in $\Psi$ and there is a $\gamma \in \mathbb{N}$ such that $\pi$ induces a subproof $\pi\downarrow \gamma$ of $\Psi \downarrow \gamma$
such that the occurrences in $\Omega$ correspond to cut-ancestors below $\pi\downarrow \gamma$~\cite{thesis2012Tsvetan,CERESS2}. Note that the set of relevant cut-configurations can be computed given a proof schema $\Psi$. To represent a proof symbol $\varphi$ and configuration $\Omega$ pairing in a clause set we assign them a {\em clause set symbol} $cl^{\varphi,\Omega}(a,\bar{x})$, where $a$ is a term of the $\omega$ sort. 

\begin{definition}[Characteristic clause term~\cite{CERESS2}]\label{def:charterm}
\index{Characteristic Term}
Let $\pi$ be an $\mathbf{LKS}$-proof and $\Omega$ a configuration. In the following, by $\Gamma_{\Omega}$ , $\Delta_{\Omega}$ and $\Gamma_{C}$ , $\Delta_{C}$ we will denote multisets of formulae of $\Omega$- and $cut$-ancestors respectively. Let $r$ be an inference in $\pi$. We define the clause-set term $\Theta_r^{\pi,\Omega}$ inductively:
\begin{itemize}
\item if $r$ is an axiom of the form $\Gamma_{\Omega} ,\Gamma_C , \Gamma \vdash \Delta_{\Omega} ,\Delta_C , \Delta$, then \\ $\Theta_{r}^{\pi,\Omega} = \left\lbrace \Gamma_{\Omega} ,\Gamma_C  \vdash \Delta_{\Omega} ,\Delta_C \right\rbrace $
\item if $r$ is a proof link of the form
\AxiomC{$\psi(a,\bar{u})$}
\dashedLine
\UnaryInfC{$\Gamma_{\Omega} ,\Gamma_C , \Gamma \vdash \Delta_{\Omega} ,\Delta_C , \Delta$}
\DisplayProof
then define $\Omega'$ as the set of formula occurrences from $\Gamma_{\Omega} ,\Gamma_C  \vdash \Delta_{\Omega} ,\Delta_C$ and $\Theta_{r}^{\pi,\Omega} = cl^{\psi,\Omega}(a,\bar{u})$
\item if $r$ is a unary rule with immediate predecessor \index{Predecessor} $r'$ , then $\Theta_{r}^{\pi,\Omega} =  \Theta_{r'}^{\pi,\Omega}$

\item if $r$ is a binary rule with immediate predecessors $r_1 $, $r_2 $, then 
\begin{itemize}
\item if the auxiliary formulae of $r$ are $\Omega$- or $cut$-ancestors, then
$\Theta_{r}^{\pi,\Omega} = \Theta_{r_1}^{\pi,\Omega} \oplus \Theta_{r_2}^{\pi,\Omega}$
\item otherwise, $\Theta_{r}^{\pi,\Omega} = \Theta_{r_1}^{\pi,\Omega} \otimes \Theta_{r_2}^{\pi,\Omega}$
\end{itemize}
\end{itemize}
Finally, define $\Theta^{\pi,\Omega} = \Theta_{r_0}^{\pi,\Omega}$ where $r_0$ is the last inference in $\pi$ and $\Theta^{\pi} = \Theta^{\pi,\emptyset}$. We call $\Theta^{\pi}$ the characteristic term of $\pi$. 
\end{definition}

Clause terms evaluate to sets of clauses by $|\Theta| = \Theta$ for clause sets $\Theta$, $|\Theta_1 \oplus \Theta_2| = |\Theta_1| \cup |\Theta_2|$, $|\Theta_1 \otimes \Theta_2| = \{C \circ D \mid C \in |\Theta_1|, D \in |\Theta_2|\}$.

The characteristic clause term is extracted for each proof symbol in a given proof schema $\Psi$, and together they make the {\em characteristic term schema} for $\Psi$.
\begin{definition}[Characteristic Term Schema\cite{CERESS2}]
Let $\Psi = \left\langle \psi_{1},\cdots, \psi_{\alpha} \right\rangle $ be a proof schema. We define the rewrite rules for clause-set symbols for all proof symbols $\psi_{\beta}$  and configurations $\Omega$ as $cl^{\psi_{\beta},\Omega}(0,\overline{u}) \rightarrow \Theta^{\pi_{\beta},\Omega}$ and $cl^{\psi_{\beta},\Omega}(k+1,\overline{u}) \rightarrow \Theta^{\nu_{\beta},\Omega}$ where $1\leq \beta\leq \alpha$. Next, let $\gamma\in \mathbb{N}$ and $cl^{\psi_{\beta},\Omega}\downarrow_{\gamma}$ be the normal form of $cl^{\psi_{\beta},\Omega}(\gamma,\overline{u})$ under the rewrite system just given extended by rewrite rules for defined function and predicate symbols. Then define $\Theta^{\psi_{\beta},\Omega} = cl^{\psi_{\beta},\Omega}$ and $\Theta^{\Psi,\Omega} = cl^{\psi_{1},\Omega}$ and finally the characteristic term schema $\Theta^{\Psi} =  \Theta^{\Psi,\emptyset}$.
\end{definition}
\subsection{Resolution Proof Schemata} \label{ResProSch}
From the characteristic clause set we can construct {\em clause schemata} which are an essential part of the definition of {\em resolution terms} and {\em resolution proof schema }\cite{CERESS2}. Clause schemata are a generalization of clauses which serve as the base for the resolution terms used to construct a resolution proof schema. Though, for the rest of this work, we leave clause schemata as a theoretical construct and work directly with meta-level clauses based on clause schemata. One additional notion needed for defining resolution proof schema is that of {\em clause variables}. The idea behind clause variables is that parts of the clauses at the leaves can be passed down a refutation to be used later on. The definition of resolution proof schemata uses clause variables as a way to handle this passage of clauses. Substitutions on clause variables are defined in the usual way. 

\begin{definition}[Clause Schema \cite{CERESS2}]
Let $b$  an  $\omega$-term, $\overline{u}$ a vector
of schematic variables and $\overline{X}$  a vector of clause variables. Then $c(b, \overline{u}, \overline{X})$ is
a clause schema w.r.t. the rewrite system $R$:
\begin{center}
\begin{tabular}{cc}
$c(0, \overline{u}, \overline{X}) \rightarrow C \circ X$ & $c(k + 1, \overline{u}, \overline{X}) \rightarrow c(k, \overline{u}, \overline{X}) \circ D$
\end{tabular}
\end{center}
where $C$ is a clause with $V(C) \subseteq \left\lbrace \overline{u} \right\rbrace$  and $D$ is a clause with $V(D) \subseteq \left\lbrace k, \overline{u}\right\rbrace $. Clauses and clause variables are clause schemata w.r.t. the empty rewrite system. Later when we introduce carriage return list, note that both the size and position in the list are $\omega$-terms and thus can be used in clause set schema.
\end{definition}
\begin{definition}[Resolution Term \cite{CERESS2}]
Clause schemata are resolution terms; if
$\rho_1$ and $\rho_2$  are resolution terms, then $r(\rho_1 ; \rho_2 ; P)$ is a resolution term, where $P$ is an atom formula schema.
\end{definition}

Essentially a resolution term  $r(\rho_1 ; \rho_2 ; P)$ is interpreted as resolving  $\rho_1, \rho_2$ on the atom $P$. The notion of most general unifier has not yet been introduced being that we introduce the concept as a separate schema from the resolution proof schema. 

\begin{definition}[Resolution Proof Schema \cite{OrderPigeonsArxivVersion,CERESS2}]
A resolution proof schema $\mathcal{R}(n)$ is a structure $( \varrho_1 , \cdots , \varrho_\alpha )$ together with a set of rewrite rules $\mathcal{R} = \mathcal{R}_1 \cup \cdots \cup \mathcal{R}_{\alpha}$ ,
where the $\mathcal{R}_i\ (for \ 1 \leq i \leq \alpha )$ are pairs of rewrite rules 
\begin{center}
\begin{tabular}{cc}
$\varrho_i (0, \overline{w},\overline{ u},  \overline{X} ) \rightarrow \eta_i$ & $\varrho_i (k+1,\overline{w},\overline{ u},  \overline{X} ) \rightarrow \eta'_i $
\end{tabular}
\end{center}

where, $\overline{w},\overline{ u},$ and $ \overline{X}$ are vectors of $\omega$, schematic, and clause variables respectively, $\eta_i$ is a resolution term over terms of the form $\varrho_j(a_j , \overline{m},\overline{ t},  \overline{C})$ for $i<j\leq \alpha$, and $\eta'_i$ is a resolution term over terms of the form $\varrho_j(a_j , \overline{m},\overline{ t},  \overline{C})$ and $\varrho_i(k, \overline{m},\overline{ t},  \overline{C})$ for $i < j \leq \alpha$; by $a_j$, we denote a term of the  $\omega$ sort.
\end{definition}

Resolution proof schema  simulates a recursive construction of a resolution derivation tree and can be unfolded into a tree once the free parameter is instantiated. The expected properties of resolution and resolution derivations hold for resolution proof schema, more detail can be found in \cite{CERESS2}. Notice that an ordering is forced on the indexing value $k$. This is where we run into problems later. 

\begin{definition}[Substitution Schema \cite{CERESS2}]
Let $u_1 , \cdots , u_{\alpha}$ be schematic variable
symbols of type $\omega \rightarrow \iota$ and $t_1 , \cdots , t_{\alpha}$ be term schemata containing no other $\omega$-variables than $k$. Then a substitution schema is an expression of the form $\left[ u_1 /\lambda k.t_{1} , \cdots , u_{\alpha} /\lambda k.t_{\alpha} \right]$.
\end{definition}

Semantically, the meaning of the substitution schema is for all $\gamma\in \mathbb{N}$ we have a substitution of the form $\left[ u_1(\gamma) /\lambda k.t_{1}\downarrow_{\gamma} , \cdots , u_{\alpha}(\gamma) /\lambda k.t_{\alpha}\downarrow_{\gamma} \right]$. For the resolution proof schema the semantic meaning is as follows, let $R(n) = ( \varrho_{1}, \cdots , \varrho_{\alpha} )$ be a resolution proof schema, $\theta$ be a clause substitution, $\nu$ an $\omega$-variable substitution, $\vartheta$ be a substitution schema, and $\gamma \in \mathbb{N}$, then  $R(\gamma)\downarrow$ denotes a resolution
term which has a normal form of $\varrho_1 (n,\overline{w}, \overline{u} , \overline{X} )\theta\nu\vartheta[n/\gamma ]$ w.r.t. $R$ extended by rewrite rules for defined function and predicate symbols.

\subsection{Herbrand Systems}\label{sec:hersys}

From the resolution proof schema and the substitution schema we can extract a so-called {\em Herbrand system}. The idea is to generalize the mid sequent theorem of Gentzen to proof schemata \cite{Baaz:2013:MC:2509679,prooftheory}. This theorem states that a proof (cut-free or with quantifier-free cuts) of a prenex end-sequent can be transformed in a way that there is a midsequent separating quantifier inferences from propositional ones. The mid-sequent is propositionally valid (w.r.t. the axioms) and contains (in general several) instances of the matrices of the prenex formulae; it is also called a {\em Herbrand sequent}. The schematic CERES method was designed such that a Herbrand system can be extracted. Our generalization preserves this property, however, the recursion for list construction must be over carriage return list (see Sec.~\ref{sec:CRLGSRC}), i.e. replace $\gamma$ by $\mathit{CR}_\gamma$ in Def.~\ref{def:herbrand}. We restrict the sequents further to skolemized ones. In the schematization of these sequents we allow only the matrices of the formulae to contain schematic variables (the number of formulae in the sequents and the quantifier prefixes are fixed).  

\begin{definition}[skolemized prenex sequent schema\cite{OrderPigeonsArxivVersion}]\label{def:sps-schema}
Let 
$S(n) = \Delta_n, \varphi_{1}(n), \cdots, \varphi_{k}(n) \vdash $ \\ $ \psi_{1}(n), \cdots,  \psi_{l}(n), \Pi_n$ ,for $ k,l\in \mathbb{N}$, where 

\begin{tabular}{ll}
$\varphi_{i}(n) = \forall x_{1}^{i}\cdots \forall x_{\alpha_{i}}^{i} F_{i}(n,x_{1}^{i},\cdots, x_{\alpha_{i}}^{i}),$ \ \ &$\psi_{j}(n) = \exists x_{1}^{j}\cdots \exists y_{\beta_{j}}^{j} E_{j}(n,y_{1}^{j},\cdots, y_{\beta_{j}}^{j}),$
\end{tabular}

for $\alpha_{i},\beta_{j}\in \mathbb{N}$, $F_{i}$ and $E_{j}$ are quantifier-free schematic formulae and $\Delta_n,\Pi_n$ are multisets of quantifier-free formulae of fixed size; moreover, the only free variable in any of the formulae is $n:\omega$. Then $S(n)$ is called a skolemized prenex sequent schema (sps-schema).
\end{definition}

\begin{definition}[Herbrand System\cite{OrderPigeonsArxivVersion}]
\label{def:herbrand}
Let $S(n)$ be a sps-schema as in Definition~\ref{def:sps-schema}. Then a Herbrand system for $S(n)$ is a rewrite system ${\cal R}$ (containing the list constructors and unary function symbols $w_{i}^{x}$, for x $\in \left\lbrace \Phi, \Psi \right\rbrace$), such that  for each $\gamma \in \mathbb{N}$, the normal form of $w_{i}^{x}(\gamma)$ w.r.t ${\cal R}$ is a list of list of terms $t_{i,x,\gamma}$ (of length $m(i,x)$) such that the sequent 
$$\Delta_\gamma,\Phi_1(\gamma),\ldots,\Phi_k(\gamma) \vdash \Psi_1(\gamma),\ldots,\Psi_l(\gamma)$$
for 
\begin{eqnarray*}
\Phi_j(\gamma) &=& \bigwedge^{m(j,\varphi)}_{p=1}E_j(\gamma,t_{j,\varphi,\gamma}(p,1),\ldots,t_{j,\varphi,\gamma}(p,\alpha_j))\ (j=1,\ldots,k),\\
\Psi_j(\gamma) &=& \bigvee^{m(j,\psi)}_{p=1}F_j(\gamma,t_{j,\psi,\gamma}(p,1),\ldots,t_{j,\psi,\gamma}(p,\beta_j))\ (j=1,\ldots,l),
\end{eqnarray*}
is \textbf{LKE}-provable.
\end{definition}

\section{A ``Mathematical'' Proof of the NiA Statement}\label{sec:MathNiA}

In this section we provide a mathematical proof of the NiA statement (Thm. \ref{thm:finalpart}). The proof is very close in structure to the formal proof written in the \textbf{LKS}-calculus, which can be found in Sec. \ref{sec:FormNiA}. We skip the basic structure of the proof and outline the structure emphasising the cuts. We will refer to the interval $\left\lbrace 0, \cdots, n-1 \right\rbrace $ as $\mathbb{N}_{n}$. Let $rr_{f}(n)$ be the following sentence, for $n\geq 2$: there exists $p,q \in \mathbb{N}$ such that  $p < q$ and  $f(p) = f(q)$,  or for all $x \in \mathbb{N}$ there exists a $y \in \mathbb{N}$ such that $x\leq y$ and $f(y)\in \mathbb{N}_{n-1}$.

\begin{lemma}
\label{lem:Inducbase}
Let $f:\mathbb{N} \rightarrow \mathbb{N}_{n}$, where $n\in \mathbb{N}$, be total, then  $rr_{f}(n)$ or there exists $p,q \in \mathbb{N}$ such that  $p < q$ and  $f(p) = f(q)$.
\end{lemma}
\begin{proof}
We can split the codomain into $\mathbb{N}_{n-1}$ and $\left\lbrace n \right\rbrace$, or the codomain is $\left\lbrace 0 \right\rbrace$.
\end{proof}
\begin{lemma}
\label{lem:inDucstep}
Let $f$ be a function as defined in Lem. \ref{lem:Inducbase} and $2< m\leq n$, then if $rr_{f}(m)$ holds so does $rr_{f}(m-1)$. 
\end{lemma}
\begin{proof}
Apply the steps of Lem. \ref{lem:Inducbase} to the right side of the {\em or} in $rr_{f}(m)$.
\end{proof}
\begin{theorem}
\label {thm:finalpart}
Let $f$ be a function as defined in Lem. \ref{lem:Inducbase} , then there exists $i,j \in \mathbb{N}$ such that $i<j$ and $f(i) = f(j)$.
\end{theorem}
\begin{proof}
Chain together the implications of  Lem. \ref{lem:inDucstep} and derive $rr_{f}(2)$, the rest is trivial by   Lem. \ref{lem:Inducbase}.
\end{proof}
This proof makes clear that the number of cuts needed to prove the statement is parametrized by the size of the codomain of the function $f$. The formal proof of the next section outlines more of the basic assumptions being that they are needed for constructing the characteristic clause set. 

\section{NiA formalized in the \textbf{LKS}-calculus}\label{sec:FormNiA}

In this section we provide a formalization of the NiA-schema whose proof schema representation is $\left\langle (\omega(0),\omega(n+1)),(\psi(0),\psi(n+1)) \right\rangle$. Cut-ancestors will be marked with a $^*$ and $\Omega$-ancestors with $^{**}$. We will make the following abbreviations: $ EQ_{f} \equiv \exists p \exists q( p < q \wedge   f(p)= f(q))$, $I(n) \equiv \forall x \exists y ( x\leq y \wedge \bigvee_{i=0}^{n} f(y) = i)$, $I_s(n) \equiv \forall x \exists y ( x\leq y \wedge f(y) = n)$ and $AX_{eq}(n) \equiv f(\beta) = n^{*} ,f(\alpha) = n^{*} \vdash  f(\beta)= f(\alpha)$ (the parts of $AX_{eq}(n)$ marked as cut ancestors are always cut ancestors in the NiA-schema).
\begin{figure}[H]
\begin{tiny}
\begin{prooftree}
\AxiomC{$\begin{array}{c} \vdash \alpha \leq \alpha^{*} \end{array}$}
\AxiomC{$\begin{array}{c}  f(\alpha) = 0  \vdash   f(\alpha) = 0^{*}\end{array}$}
\RightLabel{$\wedge :r$}
\BinaryInfC{$\begin{array}{c}\vdots\\ \forall x  f(x) = 0  \vdash I(0)^{*} \end{array}$}
\AxiomC{$\begin{array}{c}  s(\beta) \leq \alpha^{*} \vdash \beta < \alpha\end{array}$}

\AxiomC{$\begin{array}{c}   AX_{eq}(0)\end{array}$}
\RightLabel{$\wedge:r$}

\BinaryInfC{$\begin{array}{c} \vdots\\ I(0)^{*}  \vdash EQ_{f}\end{array}$}
\RightLabel{$cut$}
\BinaryInfC{$\begin{array}{c}  \forall x  f(x) = 0 \vdash  EQ_{f}\end{array}$}
\end{prooftree}
\end{tiny}
\caption{Proof symbol $\omega(0)$}
\end{figure}

\begin{figure}[H]
\begin{tiny}
\begin{prooftree}
\AxiomC{$\begin{array}{c}\varphi(n+1)  \end{array}$}
\dottedLine
\UnaryInfC{$\begin{array}{c}  I(n+1)^{*} \vdash    EQ_{f} \end{array}$}
\AxiomC{$\begin{array}{c} \vdash  \alpha \leq \alpha^{*} \end{array}$}

\AxiomC{$\begin{array}{c}  \bigvee_{i=0}^{n+1} f(\alpha) = i  \vdash  \bigvee_{i=0}^{n+1} f(\alpha) = i^{*}\end{array}$}
\RightLabel{$\wedge :r$}
\BinaryInfC{$\begin{array}{c}\vdots \\ \forall x \bigvee_{i=0}^{n+1} f(x) = i \vdash   I(n+1)^{*} \end{array}$}
\RightLabel{$cut$}
\BinaryInfC{$\begin{array}{c}  \forall x \bigvee_{i=0}^{n+1} f(x) = i \vdash    EQ_{f} \end{array}$}
\end{prooftree}

\end{tiny}
\caption{Proof symbol $\omega(n+1)$}
\end{figure}

\begin{figure}[H]
\begin{tiny}
\begin{prooftree}
\AxiomC{$\begin{array}{c}  s(\beta) \leq \alpha^{*} \vdash  \beta < \alpha\end{array}$}
\AxiomC{$\begin{array}{c}   AX_{eq}(0)\end{array}$}
\RightLabel{$\wedge:r$}

\BinaryInfC{$\begin{array}{c} \vdots\\ I_s(0)^{*} \vdash    EQ_{f}\end{array}$}
\end{prooftree}
\end{tiny}
\caption{Proof symbol $\psi(0)$}
\end{figure}

\begin{figure}[H]
\begin{tiny}
\begin{prooftree}
\AxiomC{$\begin{array}{c}   max(\alpha,\beta)\leq \gamma^{**}  \vdash \\  \alpha \leq \gamma^{*}  \end{array}$} 
\AxiomC{$\begin{array}{c}    f(\gamma) = 0^{**}     \vdash  f(\gamma) = 0^{*} \end{array}$}
\UnaryInfC{$\begin{array}{c}  \vdots \end{array}$}
\AxiomC{$\begin{array}{c}    f(\gamma) = (n+1)^{**}     \vdash \\ f(\gamma) = n+1^{*} \end{array}$}
\BinaryInfC{$\begin{array}{c}\vdots    \end{array}$}
\RightLabel{$\wedge :r$}
\BinaryInfC{$\begin{array}{c}\vdots    \end{array}$}
\AxiomC{$\begin{array}{c}    max(\alpha,\beta)\leq \gamma^{**}  \vdash \\  \beta \leq \gamma^{*}  \end{array}$}
\RightLabel{$\wedge :r$}
\BinaryInfC{$\begin{array}{c} I((n+1))^{**} \vdash     I(n)^{*},    I_{s}(n+1)^{*} \\ \vdots \end{array}$}
\end{prooftree}
\end{tiny}

\begin{tiny}
\begin{prooftree}
\AxiomC{$\begin{array}{c}\vdots \\ I(n+1)^{**} \vdash     I(n)^{*},    I_{s}(n+1)^{*}\end{array}$}
\AxiomC{$\varphi(n)$}
\dottedLine
\UnaryInfC{$\begin{array}{c}  I(n)^{*}\vdash     EQ_{f}\end{array}$}
\RightLabel{$cut$}
\BinaryInfC{$\begin{array}{c}   I(n+1)^{**} \vdash   EQ_{f}, I_{s}(n+1)^{*}\\\vdots\end{array}$}
\end{prooftree}
\end{tiny}

\begin{tiny}
\begin{prooftree}
\AxiomC{$\begin{array}{c}\vdots \\  I(n+1)^{**} \vdash   EQ_{f}, I_{s}(n+1)^{*}\end{array}$}
\AxiomC{$\begin{array}{c}  s(\beta) \leq \alpha^{*} \vdash  \beta < \alpha\end{array}$}
\AxiomC{$\begin{array}{c}   AX_{eq}(n+1)\end{array}$}
\RightLabel{$\wedge:r$}

\BinaryInfC{$\begin{array}{c} \vdots\\I_{s}(n+1)^{*}  \vdash EQ_{f} \end{array}$}
\RightLabel{$cut$}
\BinaryInfC{$\begin{array}{c}   I(n+1)^{**} \vdash    EQ_{f} , EQ_{f} \end{array}$}
\RightLabel{$c:r$}
\UnaryInfC{$\begin{array}{c} I(n+1)^{**} \vdash    EQ_{f}  \end{array}$}
\end{prooftree}
\end{tiny}
\caption{Proof symbol $\psi(n+1)$}
\end{figure}

\section{Characteristic Clause set Schema Extraction }\label{sec:CCSSE}
The outline of the formal proof provided above highlights the inference rules which directly influence the characteristic clause set schema of the NiA-schema. Also to note are the configurations of the NiA-schema which are relevant, namely, the empty configuration $\emptyset$ and a schema of configurations $\Omega(n) \equiv \forall x \exists y ( x\leq y \wedge \bigvee_{i=0}^{n} f(y) = i)$. Thus, we have the following:
\begin{equation}
\label{seq:charclaset}
\begin{array}{l} CL_{NiA}(0)\equiv \Theta^{\omega,\emptyset}(0)\equiv\left( cl^{\psi,\Omega(0)}(0)\oplus  \left\lbrace \vdash \alpha\leq \alpha \right\rbrace  \right)\oplus \left\lbrace \vdash f(\alpha)=0  \right\rbrace 
\\
 cl^{\psi,\Omega(0)}(0) \equiv\Theta^{\psi,\Omega(0)}(0) \equiv\left\lbrace  s(\beta)\leq   \alpha \vdash\right\rbrace  \otimes \left\lbrace f(\alpha)=0,  f(\beta)=0\vdash \right\rbrace 
 \\
 CL_{NiA}(n+1)\equiv \Theta^{\omega,\emptyset}(n+1)\equiv  \left( cl^{\psi,\Omega(n+1)}(n+1)\oplus \left\lbrace\vdash \alpha\leq \alpha\right\rbrace  \right)\oplus \left\lbrace \vdash \bigvee_{i=0}^{n+1} f(\alpha)=i  \right\rbrace 
\\
\begin{array}{l}  cl^{\psi,\Omega(n+1)}(n+1)\equiv \Theta^{\psi,\Omega(n+1)}(n+1)\equiv \left(  \left( cl^{\psi,\Omega(n)}(n) \oplus \left( \left\lbrace s(\beta)\leq \alpha \vdash\right\rbrace  \otimes \right. \right. \right. \\ \left. \left. \left.   \left\lbrace f(\alpha)= (n+1), f(\beta)=(n+1)\vdash\right\rbrace  \right) \right) \oplus \left\lbrace  max(\alpha,\beta)\leq \gamma \vdash \alpha \leq \gamma \right\rbrace \right)   \oplus \\ \left\lbrace max(\alpha,\beta)\leq \gamma \vdash    \beta \leq \gamma \right\rbrace 
\end{array}
\end{array}
\end{equation}

In the characteristic clause set schema  $CL_{NiA}(n+1)$ presented in Eq.\ref{seq:charclaset} tautologies are already eliminated. {\em Evaluation} of $CL_{NiA}(n+1)$ yields the following clause set $C(n)$, where $0\leq k\leq n$:
\[
\begin{array}{l}
C1 \equiv \vdash \alpha  \leq  \alpha \ \ C2 \equiv  max(\alpha, \beta) \leq \gamma \vdash \alpha \leq \gamma \ \ C3 \equiv max(\alpha, \beta)  \leq \gamma \vdash \beta \leq \gamma \\ 
C4(k)\equiv  f(\beta)  = k ,  f(\alpha)  = k  ,   s(\beta) \leq \alpha  \vdash \ \ 
C5\equiv \vdash f(\alpha) = 0 , \cdots , f(\alpha) = n
\end{array}
\]

\section{Refutation of the NiA-schema's Characteristic Clause Set Schema} \label{sec:refuteset}
In this section we provide a refutation of  $C(n)$ for every value of $n$. We prove this result by first deriving a set of clauses similar to the \textbf{TACNF} clause set of \cite{SimonThesis}; we will consider the members of this clause set the least elements of a well ordering. Then we show how resolution can be applied to this least elements to derive clauses of the form $f(\alpha)= i \vdash $ for $0\leq i \leq n$. The last step is simply to take the clause $(C5)$ from the clause set $C(n)$ and resolve it with each of the  $f(\alpha) = i\vdash $ clauses.

\begin{definition}
\index{Iterated Max term}
\label{def:maxterm}
We define the primitive recursive term $m(k,x,t)$, where $x$ is a schematic variable and $t$ a term: $\left\lbrace m(k+1,x,t) \Rightarrow  \right. $ $ \left. m(k,x,max(s(x_{k+1}),t))  \ ; \ m(0,t) \Rightarrow t \right\rbrace$.

\end{definition}

\begin{definition}
\index{Resolution Rule}
\label{def:resstep}
We define the resolution rule $res(\sigma,P)$ where $\sigma$ is a unifier and $P$ is a predicate as follows:
\begin{prooftree}
\AxiomC{$\begin{array}{c}\Pi \vdash P^*, \Delta \end{array}$}
\AxiomC{$\begin{array}{c} \Pi' , P^{**}  \vdash \Delta'  \end{array}$}
\RightLabel{$res(\sigma,P)$} 
\BinaryInfC{$\begin{array}{c}\Pi\sigma , \Pi'\sigma \vdash \Delta\sigma , \Delta'\sigma\end{array}$}
\end{prooftree}
 The predicates $P^*$ and $P^{**}$ are defined such that $P^{**}\sigma = P^*\sigma = P$. Also, there are no occurrences of $P$ in $\Pi'\sigma$ and $P$ in $\Delta\sigma$.
\end{definition}
This version of the resolution rule is not complete for unsatisfiable clause sets, it is only introduced to simplify the outline of the refutation.

\begin{lemma}
\label{lem:first}
Given $0\leq k$ and $0 \leq n$, the clause $\vdash t \leq m(k,x,t)$ is derivable by resolution from $C(n)$.
\end{lemma}
\begin{proof}
Let us consider the case when $k=0$, the clause we would like to show derivability of is $\vdash t \leq m(0,t)$, which is equivalent to the clause $\vdash t \leq t$, an instance of (C1).
Assuming the lemma holds for all $m<k+1$, we show that the lemma holds for $k+1$. By the  induction hypothesis, the instance  $\vdash max(s(x_{k+1}),t') \leq m(k,x,max(s(x_{k+1}),t'))$ is  derivable. Thus, the following derivation proves that the clause $\vdash t' \leq m(k+1,x_{k+1},t')$, where $t= max(s(x_{k+1}),t')$ for some term $t'$ is derivable:

\begin{prooftree}

\AxiomC{$\begin{array}{c}(IH) \\ \vdash P \end{array}$}
\AxiomC{$\begin{array}{c}  (C3) \\ max(\beta ,\delta ) \leq \gamma  \vdash  \delta \leq  \gamma  \end{array}$}
\RightLabel{$res(\sigma,P)$} 
\BinaryInfC{$\begin{array}{c}\vdash  t \leq m(k,x,max(s(x_{k+1}),t)) \end{array}$}
\RightLabel{$\varepsilon$}

\UnaryInfC{$\begin{array}{c}\vdash  t \leq m(k+1,x,t) \end{array}$}
\end{prooftree}
$$P = max(s(x_{k+1}),t) \leq  m(k,x,max(s(x_{k+1}),t))$$
$$\sigma =\left\lbrace \beta  \leftarrow s(x_{k+1}),  \gamma \leftarrow m(k,x,max(s(x_{k+1}),t)) ,  \delta \leftarrow t \right\rbrace $$
\\
$\square$
\end{proof}
The following corollaries follow by simple derivation.
\begin{corollary}
\label{cor:first}
Given $0\leq k,n$, the clause $ \vdash s(x_{k+1})\leq m(k,x,max(s(x_{k+1}),t))$ is derivable by resolution from $C(n)$.
\end{corollary}

\begin{corollary}
\label{cor:second}
Given $0\leq k$ and $0 \leq n$, the clause $  f(x_{k+1})= i, $\\ $f(m(k,x,max(s(x_{k+1}),t))) = i \vdash$ for $0\leq i \leq n$ is derivable by resolution from $C(n)$.
\end{corollary}

\begin{corollary}
\label{cor:third}
Given $0\leq k$ and $0 \leq n$, the clause $  f(x_{k+1})= i, f(m(k,x_{k},s(x_{k+1}))) = i \vdash$ for $0\leq i \leq n$ is derivable by resolution from $C(n)$.
\end{corollary}

\begin{definition}
\label{defclauseone}

Given $0 \leq n$, $-1\leq k \leq j \leq n$,a variable $z$, and a bijective function $b: \mathbb{N}_{n} \rightarrow \mathbb{N}_{n}$ we define the following formulae: 
\begin{equation*}
c_{b}(k,j,z) \equiv \bigwedge_{i=0}^{k} f(x_{b(i)}) = b(i) \vdash \bigvee_{i=k+1}^{j} f(m(n,x,z)) = b(i).
\end{equation*}
The formulae $c_{b}(-1,-1,z) \equiv \ \vdash$, and $c_{b}(-1,n,z) \equiv \ \vdash \bigvee_{i=0}^{n} f(z) = i$ for all values of $n$ .
\end{definition}

\begin{lemma}
\label{lem:lbase}
Given $0 \leq n$, $-1\leq k \leq n$ and for all bijective functions $b: \mathbb{N}_{n} \rightarrow \mathbb{N}_{n}$. the formula $c_{b}(k,n,z)$ is derivable by resolution from C(n).
\end{lemma}

\begin{definition}
\label{defclausetwo}
Given $0 \leq n$, $0\leq k \leq j \leq n$, and a bijective function $b: \mathbb{N}_{n} \rightarrow \mathbb{N}_{n}$ we define the following formulae: 

\begin{equation*}
c'_{b}(k,j) \equiv \bigwedge_{i=0}^{k} f(x_{i+1}) = b(i) \vdash \bigvee_{i=k+1}^{j} f(m(k,x_{k},s(x_{k+1})) = b(i).
\end{equation*}
\end{definition}
\begin{lemma}
\label{lem:lbase2}
Given $0 \leq n$, $0\leq k \leq n$ and for all bijective functions $b: \mathbb{N}_{n} \rightarrow \mathbb{N}_{n}$. the formula $c'_{b}(k,n)$ is derivable by resolution from C(n).
\end{lemma}
The proofs of Lem.~\ref{lem:lbase} \&~\ref{lem:lbase2} follow from application of Cor.\ref{cor:third} to clause $C5$. Also, the set of clauses produced by   Lem.~\ref{lem:lbase} \& ~\ref{lem:lbase2} are of similar form to the \textbf{TACNF} clause set found in Section 6 of \cite{SimonThesis}. However, we allow for a varying term structure, and thus we deviate from the precise form. Though, of most importance, is the structure of the resulting refutation, and this difference does not get in the way in our case.

\begin{definition}
\label{def:importantordering}
Given $0\leq n$ we define the ordering relation $\lessdot_{n}$ over $A_{n} = \left\lbrace (i,j) | i\leq j \right. $ $\left. \wedge 0 \leq i,j \leq n \wedge i,j \in \mathbb{N} \right\rbrace$
 s.t. for $(i,j),(l,k) \in A_n$, $(i,j) \lessdot_{n} (l,k)$ iff  $i,k,l \leq n$, $j<n$, $l\leq i$, $k\leq j$, and $i = l \leftrightarrow j \not = k$ and $j = k \leftrightarrow i \not = l$.
\end{definition}

\begin{lemma}
The ordering $\lessdot_{n}$ over $A_{n}$ for $0\leq n$ is a complete well ordering.
\end{lemma}
\begin{proof}
Every chain has a greatest lower bound\index{Greatest Lower Bound}, namely, one of the members of $A_{n}$, $(i,n)$  where $0\leq i \leq n$, and it is transitive, anti-reflexive, and anti-symmetric.
\end{proof}

The clauses proved derivable by Lem. \ref{lem:lbase2} can be paired with members of  $A_{n}$ as follows,  $c'_{b}(k,n)$ is paired with $(k,n)$. Thus, each $c'_{b}(k,n) $ is essentially the greatest lower bound of some chain in the ordering $\lessdot_{n}$ over $A_{n}$.

\begin{lemma}
\label{lem:l13}
Given  $0\leq k \leq j\leq n$, for all bijective functions $b: \mathbb{N}_{n} \rightarrow \mathbb{N}_{n}$ the clause $c'_{b}(k,j)$ is derivable from C(n).
\end{lemma}

\begin{proof}
We will prove this lemma by induction over $A_{n}$.  The base cases are the clauses $c'_{b}(k,n)$ from Lem. \ref{lem:lbase2}.  Now let us assume that the lemma holds for all clauses  $c'_{b}(k,i)$  pairs such that, $0\leq k \leq j <i \leq n$ and for all clauses $c'_{b}(w,j)$ such that $0\leq k< w \leq j\leq n$, then we want to show that the lemma holds for the clause $c'_{b}(k,j)$. We have not made any restrictions on the bijections used, we will need two different bijections to prove the theorem. The following derivation provides proof:
\begin{prooftree}
\AxiomC{$\begin{array}{c}(IH[k,j+1]) \\ \Pi_{b}(k), \vdash \Delta_{b}(k,j), P_{b}(j+1)\end{array}$}
\AxiomC{$\begin{array}{c}  (IH[k+1,k+1]) \\  \Pi_{b'}(k),  f(x_{b'(k+1)}) = b'(k+1) \vdash  \end{array}$}
\RightLabel{$res(\sigma,P)$} 
\BinaryInfC{$\begin{array}{c}   \Pi_{b}(k), \Pi_{b'}(k) \vdash \Delta_{b}(k,j) \end{array}$}
\RightLabel{$c:l$} 
\UnaryInfC{$\begin{array}{c}   \Pi_{b}(k)\vdash \Delta_{b}(k,j) \\ c'_{b}(k,j)\end{array}$}
\end{prooftree}
\begin{minipage}{.47\textwidth}
 $P_{b}(k+1) = f(m(k,x_{k},s(x_{k+1}))) = b(k+1)$,
 \end{minipage}
\begin{minipage}{.5\textwidth}
\begin{center}
$\Pi_{b}(k) \equiv \bigwedge_{i=0}^{k} f(x_{b(i)}) = b(i)$,
\end{center}
\end{minipage}
$$\Delta_{b}(k,j) \equiv \bigvee_{i=k+1}^{j} f(m(k,x_{k},s(x_{k+1}))) = b(i), $$ 
$$\sigma =\left\lbrace x_{b'(k+1)}  \leftarrow m(k,x_{k},s(x_{k+1}))\right\rbrace $$

We assume that $b'(k+1) = b(j+1)$ and that $b'(x)=b(x)$ for $0 \leq x \leq k$. 
\end{proof}

\begin{theorem}
\label{thm:refofC}
Given $n \geq 0$, $C(n)$ derives $\vdash$. 
\end{theorem}

\begin{proof}
By Lem. \ref{lem:l13}, The clauses $f(x)=  0 \vdash $ , $\cdots$ , $f(x)= n \vdash $ are derivable. Thus, we can resolve them with C5 and get $\vdash$.
\end{proof}

The reason that the above resolution refutation cannot be formalized in the resolution calculus of Sec.~\ref{ResProSch} is the necessity of allowing any bijective function to label the $\omega$-terms in Def.~\ref{defclauseone} \&~\ref{defclausetwo}. It is entirely possible, though not very likely, that another refutation avoids these issues, but given the result of \cite{SimonThesis}, this issue will eventually have to be dealt with if we want to deal with more  proof schema in terms of cut elimination. The clause set introduced   in \cite{SimonThesis}, in the schematic setting, would require any permutation of the $\omega$-terms as well. Also, such a clause set would be refuted exactly as Thm.~\ref{thm:refofC} and  Lem.~\ref{lem:l13} refute ours.  In the next section we introduce {\em carriage return list} as an alternative to indexing the resolution refutation by $\omega$-terms. 

\section{Carriage Return List and a Generalization of the Schematic Resolution Calculus}
\label{sec:CRLGSRC}
In this section we introduce carriage return list and a new schematic resolution calculus using them to index the recursion. 
\subsection{Carriage Return List}
Carriage return lists are essentially list with a pointer to an arbitrary position in the list and two operations defined for them, {\em carriage return} and {\em shift}. The carriage return operator deletes the element at the pointer and returns the pointer to the first position and the shift operator shifts the pointer to the right. The carriage return is the essential operator for formalization of the resolution refutation from the previous section because it allows us to consider an arbitrary $\omega$-term at any position in the recursion tree. 

\begin{definition}[$\omega$-list of length $n$]
A $\omega$-list of length $n$, for $n\in \mathbb{N}$, is the empty list $\Olist{\ }{\ }$ when $n=0$, or $\Olist{m}{T}$ where $m$ is an $\omega$-term and $T$ is an $\omega$-list of length $n-1$. When it is not essential we will write $\omega$-list rather than $\omega$-list of length $n$. 
\end{definition}
Given an $\omega$-list $L = \Olist{m}{T}$, $L.1 = m$ and $L.2 = T$. We will refer to the list $\Olist{n}{\Olist{n-1}{\cdots \Olist{0}{\Olist{\ }{\ }}\cdots }}$ as the canonical  $\omega$-list of length $n$. When possible we will abbreviate the list as follows $\Olist{n }{(n-1),(n-2),\cdots ,0 }$.
\begin{definition}[$\omega$-list concatenation]
Given $\omega$-lists $L$ and $H$, $L\circledcirc H$ is defined as $\left\lbrace 
L\circledcirc H =  \Olist{L.1}{ L.2 \circledcirc H} \ ; \ \Olist{\ }{\ }  \circledcirc H = H\right\rbrace 
$.
\end{definition}

\begin{definition}[$\omega$-list Length]
Given an $\omega$-lists $L$, $|L|:\omega$ is defined as 
$\left\lbrace 
|L| =  1+ |L.2| \ ;\ \right. $ \\ $\left. | \Olist{\ }{\ } |  =0
\right\rbrace $.
\end{definition}
A carriage return list is a special type of $\omega$-list.
\begin{definition}[Carriage return list]
A carriage return list $C = \CRlist{F}{m}{B}$ is an $\omega$-list of the following form $F\circledcirc \Olist{m}{B}$. Also, we define $ \CRmid{C} = m$. The canonical carriage return list  \\ $\CRlist{\ }{n}{(n-1),\cdots , 0}$, will be referred to as $\mathit{CR}_{n}$.
\end{definition}

\begin{definition}[Carriage Return List operators]
Given a carriage return list $C = \CRlist{F}{m}{B}$ we define the {\em shift} $\Rsh$ and {\em carriage return} $\Lsh$ operators as follow:
$$\begin{array}{lll}
\CRlist{F}{m}{B}\Rsh =  \CRlist{F\circledcirc \Olist{m}{\ }}{B.1}{B.2} &  \CRlist{F}{m}{\ }\Rsh  = \CRlist{F}{m}{\ } & \CRlist{\  }{\ }{ \  }\Rsh = \CRlist{\ }{\ }{\ }
\\ & & \\
\CRlist{F}{m}{B}\Lsh =  \CRlist{\  }{F.1}{F.2\circledcirc B} &  \CRlist{\ }{m}{B}\Lsh = \CRlist{\   }{B.1}{B.2} & \CRlist{\  }{\ }{ \  }\Lsh = \CRlist{\ }{\ }{\ }
\end{array}$$
Given a carriage return list $C$, the set of all derivable carriage return lists from $C$ using the above operators is $D_{C}$. 
\end{definition}  
\begin{example}
Let us consider the carriage return list $C = \CRlist{\ }{4}{3,2,1,0}$. The list $C\Rsh \Rsh \Rsh = \CRlist{4,3,2}{1}{0} $. Apply a carriage return to $C\Rsh \Rsh \Rsh$ we get $C\Rsh \Rsh \Rsh \Lsh = \CRlist{}{4}{3,2,0}$.
\end{example}
 Notice that $\CRlist{\ }{\ }{\ }$ is always derivable from a carriage return list, i.e. $C\Lsh\Lsh\Lsh\Lsh\Lsh$, and the only operator which can be applied to  $C\Rsh\Rsh\Rsh\Rsh\Rsh$ is $\Lsh$. These two special cases will replace the base case in our generalized resolution proof schema. 
\subsection{Generalized Resolution Proof Schema}
Using carriage return list we define the following resolution proof schema. 

\begin{definition}[Generalized resolution proof schema]
A generalized resolution proof schema $\mathcal{R}(n)$ is a structure $( \varrho_1 , \cdots , \varrho_\alpha )$ together with a set of rewrite rules $\mathcal{R} = \mathcal{R}_1 \cup \cdots \cup \mathcal{R}_{\alpha}$ ,
where the $\mathcal{R}_i\ (for \ 1 \leq i \leq \alpha )$ are triples of rewrite rules
\begin{center}
 
\begin{tabular}{ccc }
$\varrho_i (\CRlist{\ }{\ } {\ }, \overline{w},\overline{ u},  \overline{X} ) \rightarrow \eta_i$ & $\varrho_i (\CRlist{F}{m} {\ }, \overline{w},\overline{ u},  \overline{X} ) \rightarrow \eta_i'$ & $\varrho_i (\CRlist{F}{m} {B},\overline{w},\overline{ u},  \overline{X} ) \rightarrow \eta''_i $
\end{tabular}

\end{center}

where, $\overline{w},\overline{ u},$ and $ \overline{X}$ are vectors of $\omega$, schematic, and clause variables respectively, $\CRlist{\ }{\ } {\ },\CRlist{F}{m} {\ },\CRlist{F}{m} {B}\in D_{\mathit{CR}_{n}}$,    $\eta_i$ is a resolution term over terms of the form $\varrho_j(a_j, \overline{m},\overline{ t},  \overline{C})$ for $i<j\leq \alpha$ , $\eta'_i$ is a resolution term over terms of the form $\varrho_j(a_j , \overline{m},\overline{ t},  \overline{C})$ and $\varrho_l(\CRlist{F}{m} {\ }\Lsh, \overline{m},\overline{ t},  \overline{C})$ for $1\leq l\leq i < j \leq \alpha$, and  $\eta''_i$ is a resolution term over terms of the form $\varrho_j(a_j , \overline{m},\overline{ t},  \overline{C})$, $\varrho_l(\CRlist{F}{m}{B}\Lsh, \overline{m},\overline{ t},  \overline{C})$, and $\varrho_i(\CRlist{F}{m}{B}\Rsh, \overline{m},\overline{ t},  \overline{C})$, for $1\leq l\leq i < j \leq \alpha$; by $a_j$, we denote an arbitrary carriage return list.
\end{definition}
Notice that the previous definition of Sec~\ref{ResProSch} can be obtained from the generalized definition by ignoring the carriage return operator and ignoring the leftmost component of the rewrite system. The semantic meaning is generalized as follows, let $R(n) = ( \varrho_{1}, \cdots , \varrho_{\alpha} )$ be a resolution proof schema, $\theta$ be a clause substitution, $\nu$ an $\omega$-variable substitution, $\vartheta$ be a substitution schema, and $\gamma \in \mathbb{N}$, then  $R(\gamma)\downarrow$ denotes a resolution
term which has a normal form of $\varrho_1 (\mathit{CR}_{n},\overline{w}, \overline{u} , \overline{X} )\theta\nu\vartheta[n/\gamma ]$ w.r.t. $R$ extended by rewrite rules for defined function and predicate symbols. Essentially just exchanging the numeral in the normal form of Sec.~\ref{ResProSch} with the canonical carriage return list.
\subsection{Resolution proof schema for NiA-schema and Herbrand System}\label{genref}
We use the following abbreviations to simplify the formalization of the refutation of the NiA-schema:
\begin{center}
\begin{tabular}{cc}
$C1(t) = \vdash t\leq t $ & $C2(t,w) = \max(s(t),t) \leq w \vdash s(t)\leq w $ \\
$C3(t,w) = \max(s(t),t) \leq w \vdash t \leq w $ & $C4(t,w,n) = f(t)= n , f(w)= n, s(t)\leq w$ \\ $C5(t,n) = f(t)= 0 ,\cdots  f(t)= n $ & $m(k) = \left\lbrace\begin{array}{cc}
max(s(m(k-1)),m(k-1)) &k>0\\
0 & k=0
\end{array}  \right. $
\end{tabular}
\end{center}
Note that the refutation given in this section is not precisely the same as the refutation of Sec.~\ref{sec:refuteset}, we reordered parts of it. The major difference is we do not construct the clause set of  Lem.~\ref{lem:lbase} \&~\ref{lem:lbase2}. This was only done to show the relationship between this work and \cite{SimonThesis}. Though, the refutation has the same global structure as in Sec.~\ref{sec:refuteset}. The resolution proof schema itself has three components, of which $\rho_{1}$ simulates Thm.~\ref{thm:refofC} and Lem.~\ref{lem:l13},  $\rho_{2}$ simulates the local unification  used in Sec.~\ref{sec:refuteset} by generating all the terms at once, and $\rho_{3}$ and $\rho_{4}$  simulate Lem.~\ref{lem:first}. An example refutation can be found in Appendix \ref{exampleref}.
\begin{center}
\begin{tabular}{l}
$\rho_{1}(C,w,t,p,q,y,X,Y) \rightarrow r(\rho_{2}(\mathit{CR}_{|C|}\Rsh,\CRmid{C},1,|C\Lsh|,|C|,y,X\circ f(y_{|C|})= \CRmid{C}\vdash,Y);$ \\
\ \ \ \ \ $ \rho_{1}(C\Rsh,0,0,0,0,y,X,Y\circ \vdash f(y_{|C|})= \CRmid{C}) $\\
$\rho_{1}((C=\CRlist{F}{m}{\ }),w,t,p,q,y,X,Y) \rightarrow$ \\ 
\ \ \ \ \ $ r(\rho_{2}(\mathit{CR}_{|C|}\Rsh,\CRmid{C},1,|C\Lsh|,|C|,y,X\circ f(y_{|C|})= \CRmid{C}\vdash,Y);C5(y_{|C|})\circ Y;  f(y_{|C|})= \CRmid{C} ) $\\
$\rho_{1}(\CRlist{\ }{\ }{\ },w,t,p,q,y,X,Y) \rightarrow C5(y_{0})\circ Y$\\
\hline
$\rho_{2}(C,w,t,p,q,y,X,Y) \rightarrow$ \\ 
\ \ \ \ \ $ r(\rho_{3}(\mathit{CR}_{t},w,t,p,q,y,X\circ f(y_{p})= w\vdash,Y); \rho_{2}(C\Rsh,w,t+1,p-1,q,y,X,Y\circ \vdash f(y_{|C|})= w) $\\
$\rho_{2}((C=\CRlist{F}{m}{\ }),w,t,p,q,y,X,Y) \rightarrow $ \\ 
\ \ \ \ \  $r(\rho_{3}(\mathit{CR}_{t},w,t,p,q,y,X\circ f(y_{p})= w\vdash,Y);\rho_{1}(C\Lsh,0,0,0,0,y,X',Y\circ \vdash f(y_{|C|})= w) $\\
$\rho_{2}(\CRlist{\ }{\ }{\ },w,t,p,q,y,X,Y) \rightarrow \vdash $\\
\hline
$\rho_{3}(C,w,t,p,q,y,X,Y)\rightarrow r(C4(y_{p},y_{q},w)\circ X;r(\rho_{4}(C\Rsh ,w,t,p+1,q,y,X',Y);$\\ 
\ \ \ \ \  $ C2(y_{p},y_{q}); max(s(y_{p}),y_{p})\leq y_{q}); s(y_{p})\leq y_{q}) $\\
$\rho_{3}((C=\CRlist{F}{m}{\ }),w,t,p,q,y,X,Y)\rightarrow r(C4(y_{p},y_{q},w)\circ X;r(\rho_{4}(C ,w,t,p+1,q,y,X',Y);$\\ 
\ \ \ \ \  $ C2(y_{p},y_{q}); max(s(y_{p}),y_{p})\leq y_{q}); s(y_{p})\leq y_{q}) $\\
$\rho_{3}(\CRlist{\ }{\ }{\ },w,t,p,q,y,X,Y)\rightarrow \vdash$\\
\hline
$\rho_{4}(C,w,t,p,q,y,X,Y)\rightarrow r(\rho_{4}(C\Rsh,w,t,p+1,q,y,X,Y); C3(y_{p},y_{q});max(s(y_{p}),y_{p})\leq y_{q}) $\\
$\rho_{4}((C=\CRlist{F}{m}{\ }),w,t,p,q,y,X,Y)\rightarrow  C1(y_{q}) $\\
$\rho_{4}((C=\CRlist{\ }{\ }{\ }),w,t,p,q,y,X,Y)\rightarrow \vdash $
\end{tabular}
\end{center}
The substitution schema is $\vartheta =\left\lbrace y(k)\leftarrow \lambda k.m(k) \right\rbrace$, the clause substitution is $\theta = \left\lbrace Y\leftarrow \vdash,X\leftarrow \vdash  \right\rbrace$, and  $\omega$-variable substitution $\nu =\left\lbrace w\leftarrow 0 , t\leftarrow 0 , p\leftarrow 0, q\leftarrow 0 \right\rbrace$. Thus  it has a normal of  $\varrho_1 (\mathit{CR}_{n},w,t,p,q,y,X, Y )\theta\nu\vartheta[n/\gamma ]$. The skolemized prenex sequent schema needed for the extraction of the Herbrand system for the NiA-schema is 
 $$\varphi(n) \vdash \psi(n) \equiv \forall x \bigvee_{i=0}^n f(x)= i \vdash \exists x \exists y ( x<y \wedge f(x)=f(y)).$$
Our rewrite system  for the Herbrand system is $\mathcal{R} = \lbrace w^{\Phi}_1(\mathit{CR}_{\gamma}) , w^\Psi(\mathit{CR}_{\gamma})_1 \rbrace$, for $\gamma\in\mathbb{N}$, which are defined as follows:

\begin{tabular}{l}
\\
$ w^{\Phi}_1(C) = \left\lbrace \begin{array}{cc} 
\left[ m(|C|)\middle|\ \right]  & C=\CRlist{\ }{\ }{\ }\\
\left[ m(|C|)\middle|\ w^{\Phi}_1(C\Lsh)\right] & C=\CRlist{F }{m }{\ }\\
\left[ m(|C|)\middle| \ w^{\Phi}_1(C\Lsh)\right] & \mathrm{otherwise}
\end{array} \right. $ \\\\ $ w^{\Psi}_1(C) = \left\lbrace \begin{array}{cc} 
\left[\ \middle| \ \right]  & C=\CRlist{\ }{\ }{\ }\\
\left[ (m(\CRmid{C}),m(|C|))\middle|\ w^{\Psi}_1(C\Lsh) \right]   & C=\CRlist{F }{m }{\ }\\
\left[ (m(\CRmid{C}),m(|C|))\middle|\ w^{\Psi}_1(C\Rsh)_1\right] \circ w^{\Psi}_1(C\Lsh) & \mathrm{otherwise}
\end{array} \right. $
\end{tabular}

By $\circ$ we mean list concatenation. Note that list created by $w^{\Psi}_1(C)$ have repetition, but this is not an issue. This results in the Herbrand sequent 
$$S(n)\equiv \bigwedge_{i=0}^{n+1}\bigvee_{j=0}^n  f(m(i))= j \vdash \bigvee_{i=0}^{n}\bigvee_{j=0}^i (m(j)<m(i+1) \wedge f(m(j))= f(m(i+1))).$$

The proof of the Herbrand sequent requires the equality axiom $f(\alpha)=i,f(\beta)=i\vdash f(\alpha)=f(\beta)$ and axioms $\vdash m(i) < m(j)$, when $i<j$. 
\section{Conclusion} 

\label{sec:Conclusion}
In this work, we attempted to formalize the NiA-schema within the framework of~\cite{CERESS2}. However, the refutation of Sec. \ref{sec:refuteset} could not be formalized as resolution proof schema using the definition of Sec. \ref{ResProSch}. This was due to the refutation using every ordering of the $\omega$ sort terms. We generalized the definition of resolution proof schema, not only so the refutation of the NiA-schema can be formalized, but also to allow the formalization of the refutation structure of \cite{SimonThesis}. It happens to be the case that the clause set of the NiA-schema can nearly be transformed into a clause set resulting in the said refutation structure, i.e. the derived clause set of Lem.~\ref{lem:lbase} \&~\ref{lem:lbase2}. Also, we show that extraction of a Herbrand system from our generalized resolution proof schema is still possible, though with recursive list construction over carriage return list. In future work, we will investigate forms of schematic characteristic clause sets whose refutation is captured by our generalized definition. Also, up for investigation is finding forms of $\mathbf{LKS}$-calculus whose proof schemata result in clause sets which always have a refutation that can be formalized in our generalized definition. 

\bibliographystyle{plain}
\bibliography{references}
\appendix
\section{Missing proofs Sec.~\ref{sec:refuteset}}
\label{Appendix}

\subsection{Proof of Lem. \ref{cor:first}}
\begin{prooftree}
\AxiomC{$\begin{array}{c}(Lem. \ref{lem:first}) \\ \vdash P \end{array}$}
\AxiomC{$\begin{array}{c}  (C2) \\ max(\beta ,\delta ) \leq \gamma  \vdash  \beta \leq  \gamma  \end{array}$}
\RightLabel{$res(\sigma,P)$} 
\BinaryInfC{$\begin{array}{c}\vdash  s(x_{k+1}) \leq m(k,\overline{x},max(s(x_{k+1}),t)) \end{array}$}
\end{prooftree}
$$P = max(s(x_{k+1}),t) \leq  m(k,\overline{x},max(s(x_{k+1}),t))$$
$$\sigma =\left\lbrace \beta  \leftarrow s(x_{k+1}),  \gamma \leftarrow m(k,\overline{x},max(s(x_{k+1}),t)) ,  \delta \leftarrow t \right\rbrace $$
\\
$\square$

\subsection{Proof of Cor. \ref{cor:second}}
\begin{prooftree}
\AxiomC{$\begin{array}{c}(Cor. \ref{cor:first}) \\ \vdash P \end{array}$}
\AxiomC{$\begin{array}{c}  (C4_i) \\ f(\alpha)= i , f(\beta) = i , s(\alpha)\leq \beta  \vdash   \end{array}$}
\RightLabel{$res(\sigma,P)$} 
\BinaryInfC{$\begin{array}{c}   f(x_{k+1})= i, f(m(k,\overline{x}_{k},max(s(x_{k+1}),t))) = i \vdash \end{array}$}
\end{prooftree}
$$P = s(x_{k+1}) \leq  m(k,\overline{x}_{k},max(s(x_{k+1}),t))$$
$$\sigma =\left\lbrace \alpha  \leftarrow x_{k+1},  \beta \leftarrow m(k,\overline{x}_{k},max(s(x_{k+1}),t)) \right\rbrace $$
\\
$\square$

\subsection{Proof of Cor. \ref{cor:third}}
\begin{prooftree}
\AxiomC{$\begin{array}{c}(Lem. \ref{lem:first}) \\ \vdash P \end{array}$}
\AxiomC{$\begin{array}{c}  (C4_i) \\ f(\alpha)= i , f(\beta) = i , s(\alpha)\leq \beta  \vdash   \end{array}$}
\RightLabel{$res(\sigma,P)$} 
\BinaryInfC{$\begin{array}{c}   f(x_{k+1})= i, f(m(k,\overline{x},s(x_{k+1}))) = i \vdash \end{array}$}
\end{prooftree}
$$P = s(x_{k+1}) \leq  m(k,\overline{x}_{k},s(x_{k+1}))$$
$$\sigma =\left\lbrace \alpha  \leftarrow x_{k+1},  \beta \leftarrow m(k,\overline{x}_{k},s(x_{k+1}))) \right\rbrace $$
\\
$\square$
\subsection{Proof of Cor. \ref{lem:lbase}}
\label{sec:lbaseproof}
We prove this lemma by induction on $k$ and a case distinction on $n$.  When $n=0$ there are two possible values for $k$, $k=0$ or $k=-1$. When $k=-1$ the clause is an instance of (C5). When $k=0$ we have the following derivation:
\begin{prooftree}
\AxiomC{$\begin{array}{c}(C5) \\ c_{b}(-1,1,y) \end{array}$}
\AxiomC{$\begin{array}{c}  (Cor. \ref{cor:second}[i\leftarrow b(0), k\leftarrow 0 ]) \\ f(x_{1})= b(0), f(max(s(x_{1}),z)) = b(0) \vdash  \end{array}$}
\RightLabel{$res(\sigma,P)$} 
\BinaryInfC{$\begin{array}{c}   c_{b}(0,1,z)  \end{array}$}
\end{prooftree}
$$P = f(max(s(x_{1}),z)) = b(0)$$
$$\sigma =\left\lbrace y  \leftarrow max(s(x_{1}),z)\right\rbrace $$

By $(Cor. \ref{cor:second}[i\leftarrow b(0), k\leftarrow 0 ])$ we mean take the clause that is proven derivable by Cor. \ref{cor:second} and instantiate the free parameters of Cor. \ref{cor:second}, i.e. $i$ and $k$, with the given terms, i.e. $b(0)$ and $0$.   Remember that $b(0)$ can be either $0$ or $1$. We will use this syntax through out this section. When $n>0$ and  $k=-1$ we again trivially have (C5). When $n>0$ and  $k=0$, the following derivation suffices:

\begin{prooftree}
\AxiomC{$\begin{array}{c}(C5) \\ c_{b}(-1,n,y) \end{array}$}
\AxiomC{$\begin{array}{c}  (Cor. \ref{cor:second}[i\leftarrow b(0), k\leftarrow 0 ]) \\ f(x_{1})= b(0), f(max(s(x_{1}),z)) = b(0) \vdash  \end{array}$}
\RightLabel{$res(\sigma,P)$} 
\BinaryInfC{$\begin{array}{c}   c_{b}(0,n,z)  \end{array}$}
\end{prooftree}
$$P = f(max(s(x_{1}),z)) = b(0)$$
$$\sigma =\left\lbrace y  \leftarrow max(s(x_{1}),z)\right\rbrace $$

The main difference between the case for $n=1$ and $n>1$ is the possible instantiations of the bijection at $0$. In the case of $n>1$, $b(0) = 0 \ \vee  \cdots \vee \ b(0) = n$.  Now we assume that for all $w< k+1 <n$ and $n>0$ the theorem holds, we proceed to show that the  theorem holds for $k+1$.  The following derivation will suffice:

\begin{prooftree}
\AxiomC{$\begin{array}{c}(IH) \\ c_{b}(k,n,y) \end{array}$}
\AxiomC{$\begin{array}{c}  (Cor. \ref{cor:second}[i\leftarrow b(k+1)]) \\  f(x_{k+1})= b(k+1), P \vdash  \end{array}$}
\RightLabel{$res(\sigma,P)$} 
\BinaryInfC{$\begin{array}{c}   c_{b}(k+1,n,z)  \end{array}$}
\end{prooftree}
$$P = f(m(k,\overline{x}_{k},max(s(x_{k+1}),t))) = b(k+1)$$
$$\sigma =\left\lbrace y  \leftarrow max(s(x_{k+1}),z)\right\rbrace $$
\\
$\square$

\subsection{Proof of Lem. \ref{lem:lbase2}}
\label{sec:lbase2proof}
We prove this lemma by induction on $k$ and a case distinction on $n$.  When $n=0$ it must be the case that $k=0$. When $k=0$ we have the following derivation:
\begin{prooftree}
\AxiomC{$\begin{array}{c}(C5) \\ c_{b}(-1,0,y) \end{array}$}
\AxiomC{$\begin{array}{c}  (Cor. \ref{cor:third}[i\leftarrow 0, k\leftarrow 0 ]) \\ f(x_{1})= 0, f(s(x_{1})) = 0 \vdash  \end{array}$}
\RightLabel{$res(\sigma,P)$} 
\BinaryInfC{$\begin{array}{c}   c'_{b}(0,0)  \end{array}$}
\end{prooftree}
$$P = f(s(x_{1})) = 0$$
$$\sigma =\left\lbrace y  \leftarrow s(x_{1})\right\rbrace $$

Remember that $b(0)$ can only be mapped to  $0$. When $n>0$ and  $k=0$, the following derivation suffices:

\begin{prooftree}
\AxiomC{$\begin{array}{c}(C5) \\ c_{b}(-1,n,y) \end{array}$}
\AxiomC{$\begin{array}{c}  (Cor. \ref{cor:third}[i\leftarrow b(0), k\leftarrow 0 ]) \\ f(x_{1})= b(0), f(s(x_{1})) = b(0) \vdash  \end{array}$}
\RightLabel{$res(\sigma,P)$} 
\BinaryInfC{$\begin{array}{c}   c'_{b}(0,n)  \end{array}$}
\end{prooftree}
$$P = f(s(x_{1})) = b(0)$$
$$\sigma =\left\lbrace y  \leftarrow s(x_{1})\right\rbrace $$

The main difference between the case for $n=0$ and $n>0$ is the possible instantiations of the bijection at $0$. In the case of $n>0$, $b(0) = 0 \ \vee  \cdots \vee \ b(0) = n$.  Now we assume that for all $w\leq k$ the theorem holds, we proceed to show that the  theorem holds for $k+1$.  The following derivation will suffice:

\begin{prooftree}
\AxiomC{$\begin{array}{c}(IH) \\ c_{b}(k,n,y) \end{array}$}
\AxiomC{$\begin{array}{c}  (Cor. \ref{cor:second}[i\leftarrow b(k+1)]) \\  f(x_{k+1})= b(k+1), P \vdash  \end{array}$}
\RightLabel{$res(\sigma,P)$} 
\BinaryInfC{$\begin{array}{c}   c_{b}(k+1,n,z)  \end{array}$}
\end{prooftree}
$$P = f(m(k,\overline{x}_{k},max(s(x_{k+1}),t))) = b(k+1)$$
$$\sigma =\left\lbrace y  \leftarrow max(s(x_{k+1}),z)\right\rbrace $$
\\
$\square$

\section{Fragment of Refutation From Sec.~\ref{genref} for \textbf{n=2}}\label{exampleref}
The refutation is quite big thus we do not construct the entire refutation, but only a fragment which contains every part of the recursive structure. The substitutions are as follows: 
$$ y_{3} \leftarrow max(s(max(s(max(s(0),0)),max(s(0),0))),max(s(max(s(0),0)),max(s(0),0)))$$
$$ y_{2} \leftarrow max(s(max(s(0),0)),max(s(0),0))$$
$$ y_{1} \leftarrow max(s(0),0)$$
$$ y_{0} \leftarrow 0$$
Connections between the subproofs are marked with numbers. 

\begin{sidewaysfigure}
\begin{prooftree}
\AxiomC{$\vdash y_{1} \leq y_{1}$}
\AxiomC{$\begin{array}{c} max(s(y_{0}), y_{0}) \leq y_{1} \vdash\\  s(y_{0}) \leq y_{1} \end{array}$}

\BinaryInfC{$\vdash s(y_{0}) \leq y_{1} )$}
\AxiomC{$\begin{array}{c} f(y_{2}) = 0, \\ f(y_{1}) = 0 , \\ s(y_{0})\leq y_{1} \vdash \end{array} $}
\BinaryInfC{$\begin{array}{c}f(y_{1}) = 0,  f(y_{0}) = 0 \vdash \\ (12)\end{array}$}
\end{prooftree}
\begin{prooftree}
\AxiomC{$\begin{array}{c}(12)\\ f(y_{1}) = 0,  f(y_{0}) = 0 \vdash \end{array} $}

\AxiomC{$\begin{array}{c} \vdash f(y_{0}) = 1 \\f(y_{0}) = 2, f(y_{0}) = 0\end{array} $}

\BinaryInfC{$\begin{array}{c} f(y_{1}) = 0 \vdash  f(y_{0}) = 2,\\ f(y_{0}) = 1\end{array}$}
\AxiomC{$\begin{array}{c} \vdash  f(y_{1}) = 1, \\ f(y_{1}) = 2, f(y_{1})= 0\end{array}$}
\BinaryInfC{$\begin{array}{c}  \vdash  f(y_{1}) = 1,f(y_{0}) = 1 \\f(y_{1}) = 2, f(y_{0}) = 2 \\ (15) \end{array}$}
\end{prooftree}
\end{sidewaysfigure}

\begin{sidewaysfigure}

\begin{prooftree}
\AxiomC{$\vdash y_{3} \leq y_{3}$}

\AxiomC{$\begin{array}{c} max(s(y_{2}), y_{2}) \leq y_{3}  \vdash y_{2} \leq y_{3} \end{array}$}

\BinaryInfC{$\begin{array}{c}  \vdash y_{2} \leq y_{3} \end{array}$}

\AxiomC{$\begin{array}{c} max(s(y_{1}), y_{1}) \leq y_{3}  \vdash y_{1} \leq y_{3} \end{array}$}

\BinaryInfC{$\begin{array}{c}  \vdash y_{1} \leq y_{3} \end{array}$}

\AxiomC{$\begin{array}{c} max(s(y_{0}), y_{0}) \leq y_{3} \vdash\\  s(y_{0}) \leq y_{3} \end{array}$}

\BinaryInfC{$\vdash s(y_{0}) \leq y_{3} $}
\AxiomC{$\begin{array}{c} f(y_{3}) = 2, \\ f(y_{0}) = 2 , \\ s(y_{0})\leq y_{3} \vdash \end{array} $}
\BinaryInfC{$\begin{array}{c}f(y_{3}) = 2,  f(y_{0}) = 2 \vdash \\ (2)\end{array}$}
\end{prooftree}

\end{sidewaysfigure}

\begin{sidewaysfigure}
\begin{prooftree}
\AxiomC{$\vdash y_{2} \leq y_{2}$}

\AxiomC{$\begin{array}{c} max(s(y_{1}), y_{1}) \leq y_{2}  \vdash y_{1} \leq y_{2} \end{array}$}

\BinaryInfC{$\begin{array}{c}  \vdash y_{1} \leq y_{2} \end{array}$}

\AxiomC{$\begin{array}{c} max(s(y_{0}), y_{0}) \leq y_{2} \vdash\\  s(y_{0}) \leq y_{2} \end{array}$}

\BinaryInfC{$\vdash s(y_{0}) \leq y_{2} $}
\AxiomC{$\begin{array}{c} f(y_{2}) = 1, \\ f(y_{0}) = 1 , \\ s(y_{0})\leq y_{2} \vdash \end{array} $}
\BinaryInfC{$\begin{array}{c}f(y_{2}) = 1,  f(y_{0}) = 1 \vdash \\ (13)\end{array}$}
\end{prooftree}

\begin{prooftree}
\AxiomC{$\vdash y_{2} \leq y_{2}$}
\AxiomC{$\begin{array}{c} max(s(y_{1}), y_{1}) \leq y_{2} \vdash\\  s(y_{1}) \leq y_{2} \end{array}$}

\BinaryInfC{$\vdash s(y_{1}) \leq y_{2} )$}
\AxiomC{$\begin{array}{c} f(y_{2}) = 1, \\ f(y_{1}) = 1 , \\ s(y_{1})\leq y_{2} \vdash \end{array} $}
\BinaryInfC{$\begin{array}{c}f(y_{2}) = 1,  f(y_{1}) = 1 \vdash \\ (12)\end{array}$}
\end{prooftree}
\begin{prooftree}
\AxiomC{$\begin{array}{c}(12)\\ f(y_{2}) = 1,  f(y_{1}) = 1 \vdash \end{array} $}

\AxiomC{$\begin{array}{c} (13) \\ f(y_{2}) = 1, f(y_{0}) = 1 \vdash \end{array} $}

\AxiomC{$\begin{array}{c}(15) \\  \vdash  f(y_{1}) = 1,f(y_{0}) = 1 \\f(y_{1}) = 2, f(y_{0}) = 2\end{array} $}

\BinaryInfC{$\begin{array}{c} f(y_{2}) = 1 \vdash f(y_{1}) = 1\\ f(y_{1}) = 2, f(y_{0}) = 2 \end{array} $}
\BinaryInfC{$\begin{array}{c} f(y_{2}) = 1\vdash  f(y_{1}) = 2,\\ f(y_{0}) = 2\end{array}$}
\AxiomC{$\begin{array}{c}\vdots \\ \vdash f(y_{2}) = 2, f(y_{2}) = 1,  \\ f(y_{1}) = 2, f(y_{0}) = 2\end{array}$}
\BinaryInfC{$\begin{array}{c}  \vdash f(y_{2}) = 2, \\ f(y_{1}) = 2, f(y_{0}) = 2 \\ (11) \end{array}$}
\end{prooftree}
\end{sidewaysfigure}

\begin{sidewaysfigure}

\begin{prooftree}
\AxiomC{$\vdash y_{3} \leq y_{3}$}

\AxiomC{$\begin{array}{c} max(s(y_{2}), y_{2}) \leq y_{3}  \vdash y_{2} \leq y_{3} \end{array}$}

\BinaryInfC{$\begin{array}{c}  \vdash y_{2} \leq y_{3} \end{array}$}

\AxiomC{$\begin{array}{c} max(s(y_{1}), y_{1}) \leq y_{3}  \vdash y_{1} \leq y_{3} \end{array}$}

\BinaryInfC{$\begin{array}{c}  \vdash y_{1} \leq y_{3} \end{array}$}

\AxiomC{$\begin{array}{c} max(s(y_{0}), y_{0}) \leq y_{3} \vdash\\  s(y_{0}) \leq y_{3} \end{array}$}

\BinaryInfC{$\vdash s(y_{0}) \leq y_{3} $}
\AxiomC{$\begin{array}{c} f(y_{3}) = 0, \\ f(y_{0}) = 0 , \\ s(y_{0})\leq y_{3} \vdash \end{array} $}
\BinaryInfC{$\begin{array}{c}f(y_{3}) = 0,  f(y_{0}) = 0 \vdash \\ (10)\end{array}$}
\end{prooftree}

\end{sidewaysfigure}

\begin{sidewaysfigure}
\begin{prooftree}
\AxiomC{$\vdash y_{3} \leq y_{3}$}

\AxiomC{$\begin{array}{c} max(s(y_{2}), y_{2}) \leq y_{3}  \vdash y_{2} \leq y_{3} \end{array}$}

\BinaryInfC{$\begin{array}{c}  \vdash y_{2} \leq y_{3} \end{array}$}

\AxiomC{$\begin{array}{c} max(s(y_{1}), y_{1}) \leq y_{3} \vdash\\  s(y_{1}) \leq y_{3} \end{array}$}

\BinaryInfC{$\vdash s(y_{1}) \leq y_{3} $}
\AxiomC{$\begin{array}{c} f(y_{3}) = 0, \\ f(y_{1}) = 0 , \\ s(y_{1})\leq y_{3} \vdash \end{array} $}
\BinaryInfC{$\begin{array}{c}f(y_{3}) = 0,  f(y_{1}) = 0 \vdash \\ (9)\end{array}$}
\end{prooftree}

\begin{prooftree}
\AxiomC{$\vdash y_{3} \leq y_{3}$}
\AxiomC{$\begin{array}{c} max(s(y_{2}), y_{2}) \leq y_{3} \vdash\\  s(y_{2}) \leq y_{3} \end{array}$}

\BinaryInfC{$\vdash s(y_{2}) \leq y_{3} )$}
\AxiomC{$\begin{array}{c} f(y_{3}) = 0, \\ f(y_{2}) = 0 , \\ s(y_{2})\leq y_{3} \vdash \end{array} $}
\BinaryInfC{$\begin{array}{c}f(y_{3}) = 0,  f(y_{2}) = 0 \vdash \\ (8)\end{array}$}
\end{prooftree}
\begin{prooftree}
\AxiomC{$\begin{array}{c}(8)\\ f(y_{3}) = 0,  f(y_{2}) = 0 \vdash \end{array} $}

\AxiomC{$\begin{array}{c} (9) \\ f(y_{3}) = 0, f(y_{1}) = 0 \vdash \end{array} $}

\AxiomC{$\begin{array}{c}(10)\\ f(y_{3}) = 0, f(y_{0}) = 0 \vdash \end{array} $}
\AxiomC{$\begin{array}{c}\vdots \\  \vdash f(y_{2}) = 0, \\ f(y_{1}) = 0,\\ f(y_{0}) = 0 \end{array} $}
\BinaryInfC{$\begin{array}{c} f(y_{3}) = 0 \vdash f(y_{2}) = 0, f(y_{1}) = 0\end{array} $}

\BinaryInfC{$\begin{array}{c} f(y_{3}) = 0 \vdash f(y_{2}) = 0\end{array} $}
\BinaryInfC{$f(y_{3}) = 1\vdash $}
\AxiomC{$\begin{array}{c}\vdash f(y_{3}) = 2,f(y_{3}) = 1,\\ f(y_{3}) = 0\end{array}$}
\BinaryInfC{$\begin{array}{c} \vdash f(y_{3}) = 2 ,f(y_{3}) = 1 \\ (7)\end{array}$}
\end{prooftree}
\end{sidewaysfigure}

\begin{sidewaysfigure}

\begin{prooftree}
\AxiomC{$\vdash y_{3} \leq y_{3}$}

\AxiomC{$\begin{array}{c} max(s(y_{2}), y_{2}) \leq y_{3}  \vdash y_{2} \leq y_{3} \end{array}$}

\BinaryInfC{$\begin{array}{c}  \vdash y_{2} \leq y_{3} \end{array}$}

\AxiomC{$\begin{array}{c} max(s(y_{1}), y_{1}) \leq y_{3}  \vdash y_{1} \leq y_{3} \end{array}$}

\BinaryInfC{$\begin{array}{c}  \vdash y_{1} \leq y_{3} \end{array}$}

\AxiomC{$\begin{array}{c} max(s(y_{0}), y_{0}) \leq y_{3} \vdash\\  s(y_{0}) \leq y_{3} \end{array}$}

\BinaryInfC{$\vdash s(y_{0}) \leq y_{3} $}
\AxiomC{$\begin{array}{c} f(y_{3}) = 0 , \\ f(y_{0}) = 0 , \\ s(y_{0})\leq y_{3} \vdash \end{array} $}
\BinaryInfC{$\begin{array}{c}f(y_{3}) = 0 ,  f(y_{0}) = 0 \vdash \\ (6)\end{array}$}
\end{prooftree}

\end{sidewaysfigure}

\begin{sidewaysfigure}
\begin{prooftree}
\AxiomC{$\vdash y_{3} \leq y_{3}$}

\AxiomC{$\begin{array}{c} max(s(y_{2}), y_{2}) \leq y_{3}  \vdash y_{2} \leq y_{3} \end{array}$}

\BinaryInfC{$\begin{array}{c}  \vdash y_{2} \leq y_{3} \end{array}$}

\AxiomC{$\begin{array}{c} max(s(y_{1}), y_{1}) \leq y_{3} \vdash\\  s(y_{1}) \leq y_{3} \end{array}$}

\BinaryInfC{$\vdash s(y_{1}) \leq y_{3} $}
\AxiomC{$\begin{array}{c} f(y_{3}) = 1, \\ f(y_{1}) = 1 , \\ s(y_{1})\leq y_{3} \vdash \end{array} $}
\BinaryInfC{$\begin{array}{c}f(y_{3}) = 1,  f(y_{1}) = 1 \vdash \\ (5)\end{array}$}
\end{prooftree}

\begin{prooftree}
\AxiomC{$\vdash y_{3} \leq y_{3}$}
\AxiomC{$\begin{array}{c} max(s(y_{2}), y_{2}) \leq y_{3} \vdash\\  s(y_{2}) \leq y_{3} \end{array}$}

\BinaryInfC{$\vdash s(y_{2}) \leq y_{3} )$}
\AxiomC{$\begin{array}{c} f(y_{3}) = 1, \\ f(y_{2}) = 1 , \\ s(y_{2})\leq y_{3} \vdash \end{array} $}
\BinaryInfC{$\begin{array}{c}f(y_{3}) = 1,  f(y_{2}) = 1 \vdash \\ (4)\end{array}$}
\end{prooftree}
\begin{prooftree}
\AxiomC{$\begin{array}{c}(4)\\ f(y_{3}) = 1,  f(y_{2}) = 1 \vdash \end{array} $}

\AxiomC{$\begin{array}{c} (5) \\ f(y_{3}) = 1, f(y_{1}) = 1 \vdash \end{array} $}

\AxiomC{$\begin{array}{c}(6)\\ f(y_{3}) = 1, f(y_{0}) = 1 \vdash \end{array} $}
\AxiomC{$\begin{array}{c} \vdots \\  \vdash f(y_{2}) = 1, \\ f(y_{1}) = 1,\\ f(y_{0}) = 1 \end{array} $}
\BinaryInfC{$\begin{array}{c} f(y_{3}) = 1 \vdash f(y_{2}) = 1, f(y_{1}) = 1\end{array} $}

\BinaryInfC{$\begin{array}{c} f(y_{3}) = 1 \vdash f(y_{2}) = 1\end{array} $}
\BinaryInfC{$f(y_{3}) = 1\vdash $}
\AxiomC{$\begin{array}{c}(7) \\ \vdash f(y_{3}) = 2,f(y_{3}) = 1\end{array}$}
\BinaryInfC{$\begin{array}{c} \vdash f(y_{3}) = 2\\ (3)\end{array}$}
\end{prooftree}
\end{sidewaysfigure}

\begin{sidewaysfigure}

\begin{prooftree}
\AxiomC{$\vdash y_{3} \leq y_{3}$}

\AxiomC{$\begin{array}{c} max(s(y_{2}), y_{2}) \leq y_{3}  \vdash y_{2} \leq y_{3} \end{array}$}

\BinaryInfC{$\begin{array}{c}  \vdash y_{2} \leq y_{3} \end{array}$}

\AxiomC{$\begin{array}{c} max(s(y_{1}), y_{1}) \leq y_{3}  \vdash y_{1} \leq y_{3} \end{array}$}

\BinaryInfC{$\begin{array}{c}  \vdash y_{1} \leq y_{3} \end{array}$}

\AxiomC{$\begin{array}{c} max(s(y_{0}), y_{0}) \leq y_{3} \vdash\\  s(y_{1}) \leq y_{3} \end{array}$}

\BinaryInfC{$\vdash s(y_{0}) \leq y_{3} $}
\AxiomC{$\begin{array}{c} f(y_{3}) = 2, \\ f(y_{0}) = 2 , \\ s(y_{0})\leq y_{3} \vdash \end{array} $}
\BinaryInfC{$\begin{array}{c}f(y_{3}) = 2,  f(y_{0}) = 2 \vdash \\ (2)\end{array}$}
\end{prooftree}

\end{sidewaysfigure}

\begin{sidewaysfigure}
\begin{prooftree}
\AxiomC{$\vdash y_{3} \leq y_{3}$}

\AxiomC{$\begin{array}{c} max(s(y_{2}), y_{2}) \leq y_{3}  \vdash y_{2} \leq y_{3} \end{array}$}

\BinaryInfC{$\begin{array}{c}  \vdash y_{2} \leq y_{3} \end{array}$}

\AxiomC{$\begin{array}{c} max(s(y_{1}), y_{1}) \leq y_{3} \vdash\\  s(y_{1}) \leq y_{3} \end{array}$}

\BinaryInfC{$\vdash s(y_{1}) \leq y_{3} $}
\AxiomC{$\begin{array}{c} f(y_{3}) = 2, \\ f(y_{1}) = 2 , \\ s(y_{1})\leq y_{3} \vdash \end{array} $}
\BinaryInfC{$\begin{array}{c}f(y_{3}) = 2,  f(y_{1}) = 2 \vdash \\ (1)\end{array}$}
\end{prooftree}

\begin{prooftree}
\AxiomC{$\vdash y_{3} \leq y_{3}$}
\AxiomC{$\begin{array}{c} max(s(y_{2}), y_{2}) \leq y_{3} \vdash\\  s(y_{2}) \leq y_{3} \end{array}$}

\BinaryInfC{$\vdash s(y_{2}) \leq y_{3} )$}
\AxiomC{$\begin{array}{c} f(y_{3}) = 2, \\ f(y_{2}) = 2 , \\ s(y_{2})\leq y_{3} \vdash \end{array} $}
\BinaryInfC{$\begin{array}{c}f(y_{3}) = 2,  f(y_{2}) = 2 \vdash \\ (0)\end{array}$}
\end{prooftree}
\begin{prooftree}
\AxiomC{$\begin{array}{c}(0)\\ f(y_{3}) = 2,  f(y_{2}) = 2 \vdash \end{array} $}

\AxiomC{$\begin{array}{c} (1) \\ f(y_{3}) = 2, f(y_{1}) = 2 \vdash \end{array} $}

\AxiomC{$\begin{array}{c}(2)\\ f(y_{3}) = 2, f(y_{0}) = 2 \vdash \end{array} $}
\AxiomC{$\begin{array}{c}(11) \\  \vdash f(y_{2}) = 2, \\ f(y_{1}) = 2,\\ f(y_{0}) = 2 \end{array} $}
\BinaryInfC{$\begin{array}{c} f(y_{3}) = 2 \vdash f(y_{2}) = 2, f(y_{1}) = 2\end{array} $}

\BinaryInfC{$\begin{array}{c} f(y_{3}) = 2 \vdash f(y_{2}) = 2\end{array} $}
\BinaryInfC{$f(y_{3}) = 2\vdash $}
\AxiomC{$\begin{array}{c}(3) \\ \vdash f(y_{3}) = 2\end{array}$}
\BinaryInfC{$\vdash$}
\end{prooftree}
\end{sidewaysfigure}

\end{document}